\documentclass[11pt,a4paper]{article}

\usepackage[utf8]{inputenc}
\usepackage{amssymb,amsmath,amsthm}
\usepackage{mathrsfs}
\usepackage[T1]{fontenc}
\usepackage{graphicx}
\usepackage{graphics}
\usepackage{hyperref}
\usepackage{verbatim}
\usepackage{times}
\usepackage{verbatim}
\usepackage{pst-tree}
\usepackage{pst-all}
\usepackage{pst-pdf}
\usepackage{array}
\usepackage{listings}
\usepackage{textcomp}
\usepackage{pifont}
\usepackage{subfigure}
\usepackage{tocbibind}
\usepackage{url}
\usepackage{caption}
\usepackage{cite}
\usepackage{fancyhdr}
\usepackage{algorithm}
\usepackage[noend]{algorithmic}
\usepackage{longtable}
\usepackage{pdftricks}
\usepackage{multirow}
\usepackage{tikz}
\tikzstyle{_vertex}=[fill=white, circle,minimum size=12pt,inner sep=1pt]
\tikzstyle{_blackv}=[fill=black, circle,minimum size=8pt,inner sep=1pt]
\tikzstyle{_bigvertex}=[fill=white, circle,minimum size=21pt,inner sep=1pt]
\tikzstyle{_arc}=[->, >=stealth]

\begin{psinputs}
   \usepackage{pstricks}
   \usepackage{pst-all}
   \usepackage{multido}
\end{psinputs}

\hypersetup{colorlinks,%
            citecolor=black,%
            filecolor=black,%
            linkcolor=black,%
            urlcolor=blue}

\newcommand{\IN}{{\mathbb N}}

\newcommand{\IP}{{\mathbb P}}

\newcommand{\ssi}{\Leftrightarrow}

\newcommand{\suchas}{~\big|~}
\newcommand{\ie}{{\emph{i.e., }}}

\newcommand{\DG}{\textsf{Digenes}}

\newcommand{\nor}[1]{\widehat{#1}}

\newcommand{\cnn}[1]{\mathcal{C}_{{#1}}}
\newcommand{\knn}[1]{\mathcal{K}_{{#1}}}

\newcommand{\bnn}[1]{\mathcal{B}_{{#1}}}
\newcommand{\tnn}[1]{\mathscr{T}_{{#1}}}
\newcommand{\gn}{\ensuremath{\mathcal{G}_n}}
\newcommand{\cn}{\ensuremath{\mathcal{C}_n}}
\newcommand{\kn}{\ensuremath{\mathcal{K}_n}}
\newcommand{\pn}{\ensuremath{\mathcal{P}_n}}
\newcommand{\tn}{\ensuremath{\mathscr{T}_n}}
\newcommand{\bn}{\ensuremath{\mathcal{B}_n}}
\newcommand{\hnk}{\ensuremath{{\mathcal{H}_n(k^*)}}}
\newcommand{\hnkk}[2]{\ensuremath{{\mathcal{H}_{#1}(#2)}}}
\newcommand{\hntka}{\mathcal{H}_{n}(T,k,a)}

\newcommand{\inv}{\ensuremath{\mathcal{I}}}
\newcommand{\lset}{{\mathscr{L}}}
\newcommand{\bset}{{\mathscr{B}}}
\newcommand{\posm}[2]{\ensuremath{{\sf P}^{-}_{#1}(#2)}}
\newcommand{\posq}[2]{\ensuremath{{\sf P}^{/}_{#1}(#2)}}

\newtheorem{thm}{Theorem}
\newtheorem{conj}[thm]{Conjecture}

\newtheorem{lem}[thm]{Lemma}
\newtheorem{nott}[thm]{Notation}
\newtheorem{deff}[thm]{Definition}

\newtheorem{pbm}[thm]{Problem}
\theoremstyle{definition}
\newtheorem{remark}{Remark}

\floatname{algorithm}{Algorithm}
\algsetup{indent=2em}

\let\mylistof\listof
\renewcommand\listof[2]{\mylistof{algorithm}{Algorithm list}}

\makeatletter
\providecommand*{\toclevel@algorithm}{0}
\makeatother

\begin{document}

\title{On price of symmetrisation}
\author{Romain Absil\footnote{Algorithms Lab, Universit\'e de  Mons, Place du parc 20, B-7000 Mons, Belgium.}\ $^{,}$\footnote{FRIA grant holder. Corresponding author. E-mail: {\tt romain.absil@umons.ac.be}.}  \and 
Hadrien M\'elot\footnotemark[1]
}

\maketitle
\vspace*{0.2cm}

\hrule
\vspace*{0.2cm}
\small
\noindent
\textbf{Abstract.} We introduce the \emph{price of symmetrisation}, a concept that aims to compare fundamental differences (gap and quotient) between values of a given graph invariant for digraphs and the values of the same invariant of the symmetric versions of these digraphs. Basically, given some invariant our goal is to characterise digraphs that maximise price of symmetrisation. In particular, we show that for some invariants, as diameter or domination number, the problem is easy. 

The main contribution of this paper is about (partial) results on the price of symmetrisation of the average distance. It appears to be much more intricate than the simple cases mentioned above. First, we state a conjecture about digraphs that maximise this price of symmetrisation. Then, we prove that this conjecture is true for some particular class of digraphs (called bags) but it remains open for general digraphs. Moreover, we study several graph transformations in order to remove some configurations that do not appear in the conjectured extremal digraphs.

\vspace*{0.2cm}

\vspace*{0.2cm}
\hrule

\normalsize


\section{Introduction}\label{sec:intro}

It is a common question in optimisation and computer science to wonder how an optimal solution would change should we ensure a constraint to hold (or conversely, if we relax a constraint). This can be interpreted by a difference or a ratio between two solutions of related problems. Conceptually, it can be viewed as the price to pay to add or remove a constraint.  For example, approximation factors denote the price of polynomial worst case complexity for optimisation problems \cite{Vazirani01}. Let us also mention \emph{competition ratio}~\cite{Borodin98} in online algorithms, \emph{price of anarchy}~\cite{Koutsoupias99} in game theory and \emph{price of connectivity}~\cite{Cardinal08} in graph theory.

We introduce here a concept that also associate an estimation of the price of adding a constraint, in the context of digraphs. More precisely, we define the \emph{price of symmetrisation} for a digraph $G$ as the gap (or the quotient) between the value of a given graph invariant for a digraph $G$ and the value of the same invariant for the symmetric version of $G$, that is the graph obtained from $G$ by adding an arrow $(j, i)$ (if not already present) when an arrow $(i,j)$ exists in $G$. In other words, given an invariant $\inv$, the price of symmetrisation associated to $\inv$ expresses how would the values of $\inv$ evolve if we force digraphs to be symmetric. A natural motivation in studying prices of symmetrisation is that the symmetric version of a digraph $G$ is conceptually equivalent to the \emph{undirected} one.

In particular, we are interested in extremal problems about prices of symmetrisation~: what (family of) digraphs maximise the price of symmetrisation of a given invariant? On the contrary, literature usually consider that question the other way around, stated as an orientation problem : given an undirected graph $G$ and an invariant \inv, what is the minimum value of $\inv$ among all orientations of $G$. This question was first introduced by Chv\'atal and Thomassen~\cite{Chvatal75} about radius and diameter. Later, many papers \cite{dankelmann04,qian09,koh96,hassin89,gutin02} followed, in this search of a directed graph "almost preserving the value of \inv". Further results about orientations and this question in particular can be found in the book of Bang-Jensen and Gutin~\cite{Bang01}. We are more concerned about finding, among all graphs, what is the directed graph that will make the value of \inv~differ the most from the undirected one. 

After fixing some notations in Section~\ref{sec:not-def}, we give some (easy) solutions to the price of symmetrisation problem in the case of the diameter and the domination number in Section~\ref{sec:price}. The main part of this paper is Section~\ref{sec:avg}, devoted to the study of digraphs that have the maximum price of symmetrisation for the average distance. Experiments with the conjecture-making system \DG\cite{DG} has led to Conjecture~\ref{conj:price-trans} presented later. This conjecture describe, for a given order, which digraph is supposed to be extremal for the price of symmetrisation of the average distance. Despite many efforts, we were not able to prove this conjecture in full generality. However, we show in Section~\ref{subsec:bags} that among some digraphs (called \emph{bags}), this conjecture is true. Moreover, we give results about graph transformations in Sections~\ref{subsec:kill-c2} and \ref{subsec:contr-alg} that show that extremal graphs cannot contain induced $\cnn{2}$ in many cases. We end this section by pointing out paths of research that, if successful, will allow to prove the conjecture in full generality.

\section{Notations and definitions}\label{sec:not-def}

We devote this section to basic definitions and notations used throughout the paper and assume the reader to be familiar with basic notions of graph theory. Otherwise, we refer to the book Bang-Jensen and Gutin~\cite{Bang01} for more details about directed graphs (\emph{digraphs}). In this paper, we only consider digraphs. By abuse of terminology, we will often refer them as graphs.

Let $G = (V, A)$ a simple digraph with vertex set $V$ and arrow set $A$. We denote by $\gn$ the space of all simple non isomorphic digraphs of order $n$. A (graph) \emph{invariant} is a numerical value preserved by isomorphism, such as chromatic number, independence number, diameter and so on. We note $G \simeq H$ if graphs $G$ and $H$ are isomorphic.

We note $|x,y|_G$ the distance between two vertices $x$ and $y$ of a graph $G$. When the context is clear, we will simply note this distance $|x,y|$. Recall that the \emph{diameter} $D(G)$ of a digraph $G$ is the length (in number of arrows) of the longest shortest path of $G$. Moreover, the \emph{transmission} $\sigma(G)$ of $G$ is the sum of all lengths of all shortest paths of $G$. The \emph{average distance} $\mu(G)$ of $G$ is the arithmetic mean of these lengths, \ie $\mu(G) = \frac{\sigma(G)}{n(n-1)} \cdot$ For both these invariants, we assume underlying graphs $G$ to be \emph{strongly connected}, \ie there exists a path between any pair of vertices of $G$. We call an arrow of $G$ a \emph{bridge} if its deletion breaks the strong connectivity of $G$.

On the other hand, we say that a vertex $i$ \emph{dominates} a vertex $j$ if $(i,j) \in A$, and call a \emph{dominating set} $D$ of $G$ a set of vertices such as every vertex $j$ of $V-D$ is dominated by a vertex of $D$. The domination number $\gamma(G)$ of a digraph $G$ is the minimum size of a dominating set of $G$.

Given a digraph $G = (V,A)$, we note $\nor{G}$ the symmetric version of $G$, that is, $\nor{G} = (V, A')$, with $(i,j) \in A'$ if and only if $(i,j) \in A$ or $(j,i) \in A$, for all $i,j \in V$.

A \emph{tournament} is a digraph such as between every pair of vertices $i$ and $j$, there is either an arrow $(i,j)$ or an arrow $(j,i)$ (but not both). We note \cn, \pn~and \kn~the directed cycle, directed path and complete digraph of order $n$, respectively. We note \tn~the set of tournaments of order $n$.

\section{Basics on price of symmetrisation}\label{sec:price}

We define here the \emph{price of symmetrisation} for a digraph $G$. More formally, we define two following types of prices of symmetrisation involving an invariant \inv:
$$\posm{\inv}{G} = | \inv(G) - \inv(\nor{G}) |,$$
and
$$\posq{\inv}{G} = \frac{\inv(G)}{\inv(\nor{G})}.$$

We are then interested in the following extremal problem.

\begin{pbm}\label{pbm:sym}
Let $\inv$ be a graph invariant, what are the digraphs $G$ of order $n$ maximising or minimising  $\posm{\inv}{G}$ or $\posq{\inv}{G}$?
\end{pbm}

We will often only consider maximisation problems since in both cases any symmetric graph has a minimum price of symmetrisation\footnote{We assume $\inv(G) \geqslant 0$ for any graph $G$.}. More particularly,  if $G$ is symmetric, then $\posm{\inv}{G} = 0$ and $\posq{\inv}{G} = 1$.

Note that our definition of the price of symmetrisation is close to a problem introduced by Chv\'atal and Thomassen~\cite{Chvatal75} when $\inv$ is either the radius or the diameter. Indeed, they show that every undirected bridgeless graph of radius $r$ has an orientation of radius at most $r^2 + r$. They also work on a similar problem with diameter. However, our approach is different since we fix the order $n$ and and let the value of $\inv$ free. On the other hand, Chv\'atal and Thomassen work with a fixed value of invariant and let the order free. 
 


As we could have guessed, some prices of symmetrisation are easy. More formally, let \inv\ an invariant and $G^*$ of order $n$ be a graph such that 
$
\inv(G^*) \geqslant \inv(G),
$
and
$
\inv(\nor{G^*}) \leqslant \inv(\nor{G}),
$
for all digraphs $G$ of order $n$. Then, 
$$
\posm{\inv}{G^*} \geqslant \posm{\inv}{G},
$$
and 
$$
\posq{\inv}{G^*} \geqslant  \posq{\inv}{G}.
$$

While in general it is not the case, the two following examples illustrate easy price of symmetrisation problems for standard graph invariants, namely diameter and domination number.

\begin{pbm}\label{pbm:sym-diam}
What are the digraphs $G \in \gn$ maximising $\posm{D}{G}$ and $\posq{D}{G}$?
\end{pbm}

\begin{pbm}\label{pbm:sym-dom}
What are the digraphs $G \in \gn$ maximising $\posm{\gamma}{G}$ and $\posq{\gamma}{G}$?
\end{pbm}

In Problem \ref{pbm:sym-diam}, we of course assume the considered graphs to be strongly connected. For this problem,  as well as for the rest of the document, we need to introduce the notion of \emph{backward tournament}.

\begin{deff}\label{def:backt}
We define the \emph{backward tournament} \bn\ as the digraph $(V,A)$ of order $n \geqslant 3$ such as
\begin{eqnarray*}
V &=& \{v_1, v_2, \dots , v_n\},\\
A &=& B \cup C,\\
B &=& \left\{(v_i, v_{i+1}) \suchas 1 \leqslant i \leqslant n-1 \right\},\\
C &=&\left\{(v_i, v_j) \suchas 3 \leqslant i \leqslant n ~\wedge~ 1 \leqslant j \leqslant i-2 \right\}.
\end{eqnarray*}
\end{deff}

We note that this graph is always a strongly connected tournament, has an hamiltonian shortest path from $v_1$ to $v_n$, and is unique at fixed order $n$. Figure \ref{fig:tourn} illustrates $\bnn{6}$.

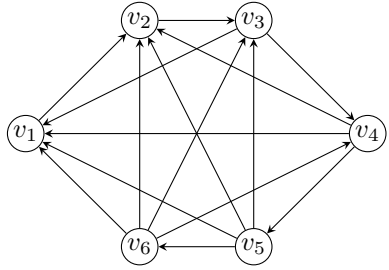
\begin{figure}[!htpd]
\begin{center}
\begin{tikzpicture}
\node[draw, _vertex] (v1) at (0, 1.5) {\small $v_1$};
\node[draw, _vertex] (v2) at (1.5, 3) {\small $v_2$};
\node[draw, _vertex] (v3) at (3, 3) {\small $v_3$};
\node[draw, _vertex] (v4) at (4.5, 1.5) {\small $v_4$};
\node[draw, _vertex] (v5) at (3, 0) {\small $v_5$};
\node[draw, _vertex] (v6) at (1.5, 0) {\small $v_6$};
\foreach \i/\j in {1/2, 2/3, 3/4, 4/5, 5/6} {\draw[_arc] (v\i) -- (v\j);}
\foreach \i/\j in {6/4, 6/3, 6/2, 6/1, 5/3, 5/2, 5/1, 4/2, 4/1, 3/1} {\draw[_arc] (v\i) -- (v\j);}
\end{tikzpicture}
\end{center}
\caption{The tournament $\bnn{6}$.}
\label{fig:tourn}
\end{figure}

It well known that, for any graph $G$ of order $n$, we have $1 \leqslant D(G) \leqslant n-1$. Moreover, we notice that $D(\bn)=n-1$ and $D(\nor{\bn})=1$, making the problem of price of symmetrisation involving $D$ easy. Observe that other graphs have such values for the diameter. Indeed, adding arrows $(v_i,v_{i-1})$ in \bn\ does not alter the diameter. We note $\bset_n$ the set of all these graphs.

The following theorem states that, among others, backward tournaments of order $n$ maximise price of symmetrisation involving diameter.

\begin{thm}\label{thm:price-diam}
Let $G \in \gn$ with $n \geqslant 3$. Then, 
\begin{eqnarray}
\label{eq:pr-diam-diff} \posm{D}{G} & \leqslant & n - 2, \\
\posq{D}{G} & \leqslant & n - 1.
\end{eqnarray}
Moreover, equalities hold if and only if $G \in \bset_n$.
\end{thm}

\begin{proof}
This proof only considers $\posm{D}{G}$, since the proof involving quotient is quite similar.

Assume $G \in \bset_n$. It is obvious that $G$ maximises the price of symmetrisation, and so that $G$~reaches Inequality \eqref{eq:pr-diam-diff}. Indeed, diameter in $G$~is maximum with a value of $n - 1$, while it is minimum in $\nor{G}$ with a value of 1. 

Let us show now that if a graph $G=(V,A)$ is such as $D(G) = n-1$ and $D(\nor{G}) = 1$, then $G \in \bset_n$. Since $D(G) = n-1$, then there is an elementary shortest path of length $n-1$ in $G$. Let us number this path vertices from $v_1$ to $v_n$, so we have $|v_1,v_n|=n-1$. Clearly, $(v_1,v_n) \notin A$ since otherwise $|v_1,v_n|=1$.

In addition, as $D(\nor{G}) = 1$, there is at least one arrow $(v_i,v_j)$ or $(v_j,v_i)$ between each pair of vertices $v_i$ et $v_j$. However, there can be no arrow $(v_i,v_j)$ with $j>i+1$ since such an arrow would reduce the length of the path from $v_1$ to $v_n$. We can so conclude all of these arrows are of type $(v_i,v_j)$ with $j < i$, and so that $G \in \bset_n$.
\end{proof}

The following notations are useful for solving Problem \ref{pbm:sym-dom}. Let $S_n$ be an undirected star of order $n$, we note $I\!S_n$ an orientation of $S_n$ where every edge of $S_n$ is oriented from its outline to its centre. Moreover, we call this graph an \emph{in-star} of order $n$. The graph on the left of Figure \ref{fig:lset} illustrates an in-star of order 5.


We note $\lset_k(G)$ the set of non isomorphic digraphs that can be obtained from a digraph $G$ by replacing at most $k$ arrows $(u,v)$ by either an arrow $(v,u)$ of by both arrows $(u,v)$ and $(v,u)$. Figure \ref{fig:lset} illustrates $\lset_1(I\!S_5)$.

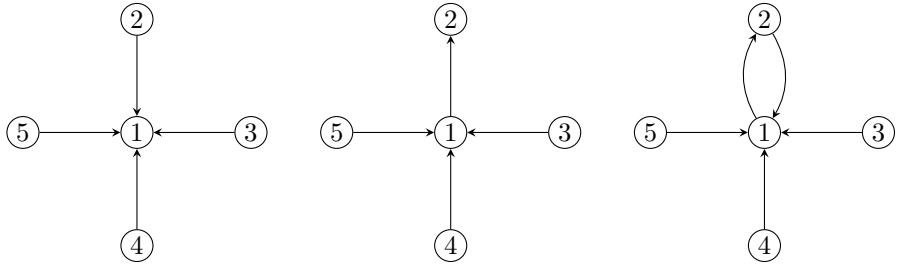
\begin{figure}[!htpd]
\begin{center}
\begin{tikzpicture}
\node[draw, _vertex] (v1) at (1.5, 1.5) {\small $1$};
\node[draw, _vertex] (v2) at (1.5, 3) {\small $2$};
\node[draw, _vertex] (v3) at (3, 1.5) {\small $3$};
\node[draw, _vertex] (v4) at (1.5, 0) {\small $4$};
\node[draw, _vertex] (v5) at (0, 1.5) {\small $5$};
\foreach \i/\j in {2/1, 3/1, 4/1, 5/1} {\draw[_arc] (v\i) -- (v\j);}
\end{tikzpicture} \hspace{0.5cm}
\begin{tikzpicture}
\node[draw, _vertex] (v1) at (1.5, 1.5) {\small $1$};
\node[draw, _vertex] (v2) at (1.5, 3) {\small $2$};
\node[draw, _vertex] (v3) at (3, 1.5) {\small $3$};
\node[draw, _vertex] (v4) at (1.5, 0) {\small $4$};
\node[draw, _vertex] (v5) at (0, 1.5) {\small $5$};
\foreach \i/\j in {1/2, 3/1, 4/1, 5/1} {\draw[_arc] (v\i) -- (v\j);}
\end{tikzpicture} \hspace{0.5cm}
\begin{tikzpicture}
\node[draw, _vertex] (v1) at (1.5, 1.5) {\small $1$};
\node[draw, _vertex] (v2) at (1.5, 3) {\small $2$};
\node[draw, _vertex] (v3) at (3, 1.5) {\small $3$};
\node[draw, _vertex] (v4) at (1.5, 0) {\small $4$};
\node[draw, _vertex] (v5) at (0, 1.5) {\small $5$};
\foreach \i/\j in {3/1, 4/1, 5/1} {\draw[_arc] (v\i) -- (v\j);}
\foreach \i/\j in {1/2, 2/1} {\draw[_arc] (v\i) to[bend left] (v\j);}
\end{tikzpicture}
\end{center}
\caption{Illustration of $\lset_1(I\!S_5)$.}
\label{fig:lset}
\end{figure}

\begin{thm}
Let $G$ a digraph of order $n \geqslant 3$, we have:\\
\begin{eqnarray}
\label{eq:pr-dom-diff} \posm{\gamma}{G} & \leqslant & n - 2, \\
 \posq{\gamma}{G} & \leqslant & n - 1.
\end{eqnarray}
Moreover, equalities hold if and only if $G \in \lset_1(I\!S_n)$.
\end{thm}

\begin{proof} Again we only deal with $\posm{\gamma}{G}$. Clearly, Inequality \eqref{eq:pr-dom-diff} is verified since graphs of $\lset_1(I\!S_n)$ maximise price of symmetrisation involving domination number. Indeed, $\gamma$ is maximum for graphs in $\lset_1(I\!S_n)$ : since at most one arrow of $I\!S_n$ has been replaced by either a reversed arrow or a $\cnn{2}$, we still need to pick all vertices of the star but its centre to build a (minimum) dominating set. In the symmetric version of these graphs, domination number is always minimum since it is enough to pick the star centre as dominating set.

Let us show now that if equality in Inequation~\eqref{eq:pr-dom-diff} holds, then $G \in \lset_1(I\!S_n)$. The only case when it happens is when $\gamma(G) = n-1$ and $\gamma(\nor{G}) = 1$. As $\gamma(G) = n-1$, no vertex can have an out-degree higher than $1$. Indeed, if such a vertex existed, then it would dominate more than one vertex (and itself) and so $\gamma(G)$ would be lower than $n-1$. Moreover, as $\gamma(\nor{G}) = 1$, there must be a vertex $v$ in $G$ such as for every vertex $w \neq v$ in $G$, there is an arrow $(v,w)$ or $(w,v)$ or both arrows $(v,w)$ and $(w,v)$. Combining this statement with the previous one, there can be only one such vertex. These conditions lead to the conclusion that $G \in \lset_1(I\!S_n)$.
\end{proof}


\section{Price of symmetrisation and average distance} \label{sec:avg}

As stated in the beginning of Section \ref{sec:price}, in the general case of prices of symmetrisation, there is no graph $G$ maximising $\inv(G)$ while minimising $\inv(\nor{G})$ at the same time, for instance with average distance (or transmission). In this case, it is well known that while \cn~maximises $\mu$ and $\sigma$ \cite{doyle02}, the graph $\kn$~minimises them. 

The following section introduces the problem and describes progress made in order to prove Conjecture \ref{conj:price-trans}, stated later.

\begin{pbm}
What are the $G \in \gn$ maximising $\posm{\mu}{G}$ or $\posq{\mu}{G}$ ? 
\end{pbm}

We consider here only $\posm{\mu}{G}$, and as for diameter assume considered graphs to be strongly connected. Moreover, since $\mu$ and $\sigma$ are the same invariant up to a $\frac{1}{n(n-1)}$ factor, we will only consider $\sigma$ to avoid some fractions in computations. 

\begin{deff}
A \emph{bag} is a tournament $T$ of order $\geqslant 3$ in which you duplicated an arrow and replaced the copy by a path $P$ of arbitrary length $\geqslant 2$. Figure \ref{fig:bag} illustrates an example of bag.

Let $T \in \tnn{k}$ with $3 \leqslant k  \leqslant n-1$ and $a \in A(T)$ , we note $\hntka$ the bag of order $n$ obtained when duplicating the arrow $a$ of $T$ and replacing it by a path of length $n-k+1$. 

More particularly, if $T \simeq \bnn{k}$, when using the same notations as Definition \ref{def:backt}, we note $\hnkk{n}{k}$ the bag of order $n$ obtained when duplicating arrow $(v_k,v_1)$ and replacing it by a path of length $n-k+1$. The graph illustrated in Figure \ref{fig:bag} is $\hnkk{8}{4}$.

\end{deff}

The following conjecture intuitively means it is either a cycle or a bag maximising price of symmetrisation, where the tournament size is roughly $40\%$ of the bag's order. This conjecture was output by the system \DG\cite{DG}.

\begin{conj}\label{conj:price-trans}
Let $G \in \gn$ strongly connected, for $2 \leqslant n \leqslant 10$ we have 
\begin{equation*}
\posm{\sigma}{G} \leqslant \posm{\sigma}{\cn},
\end{equation*}
with equality if and only if $G \simeq \cn$. Moreover, for $n \geqslant 11$, we have
\begin{equation*}
\posm{\sigma}{G} \leqslant \posm{\sigma}{\hnk},
\end{equation*}
where
\begin{equation*}
k^* = \max (\posm{\sigma}{\hnkk{n}{\lfloor r \rfloor)}}, \posm{\sigma}{\hnkk{n}{\lceil r \rceil)}}),
\end{equation*}
\begin{equation*}
r = n (\sqrt{2} - 1) + 8 - \frac{11\sqrt{2}}{2},
\end{equation*}
with equality if and only if $G \simeq \hnk$. 
\end{conj}

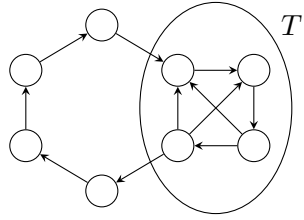
\begin{figure}[!htpd]
\begin{center}
\begin{tikzpicture}
\node[draw, _vertex] (v1) at (1, -0.6) {};
\node[draw, _vertex] (v2) at (0, 0) {};
\node[draw, _vertex] (v3) at (0, 1) {};
\node[draw, _vertex] (v4) at (1, 1.6) {};
\node[draw, _vertex] (v5) at (2, 1) {};
\node[draw, _vertex] (v6) at (3, 1) {};
\node[draw, _vertex] (v7) at (3, 0) {};
\node[draw, _vertex] (v8) at (2, 0) {};
\foreach \i/\j in {1/2, 2/3, 3/4, 4/5, 8/5, 8/1, 5/6, 6/7, 7/8, 7/5, 8/6} {\draw[_arc] (v\i) -- (v\j);}
\draw (2.5,0.5) ellipse (1 and 1.4);
\draw (3.5, 1.6) node {$T$};
\end{tikzpicture}
\end{center}
\caption{An example of bag with tournament $T$ of order $4$ and a path of length $5$.}
\label{fig:bag}
\end{figure}


A natural idea in an attempt to prove this conjecture is to proceed by graph transformations, \ie find a sequence of graph transformations that will eventually end up on \hnk, while always increasing price of symmetrisation. 

The rest of this document is devoted to this transformations approach, and organised in the following way. Firstly, Section \ref{subsec:bags} states that among all bags, $\hnk$ has a maximum price of symmetrisation. In this section, we also deal with the case of cycles when $n  \leqslant 10$. This way, we will be directly able to conclude that Conjecture \ref{conj:price-trans} is true if some transformation sequence ends up on some bag while increasing price of symmetrisation. Section \ref{subsec:kill-c2} proposes various graph transformations in order to remove induced $\cnn{2}$ from most of graphs $G$ while increasing $\posm{\sigma}{G}$. Indeed, there is no such configuration in the assumed extremal graph. Finally, Section \ref{subsec:contr-alg} proposes a simple transformation ending up either on a bag or on some other particular configuration while increasing price of symmetrisation.

The following notation will be used in these three sections and is useful to easily compute $\sigma$ in some graphs.

\begin{nott}
Let $G=(V,A) \in \gn$ strongly connected, $X,Y \subseteq V$ disjoints non empty and $v \in V$, we note\\
\begin{minipage}{6cm}
\begin{eqnarray*}
\sigma_G(X) &=& \displaystyle \sum_{x,y \in X} |x,y|_G,  \\
\sigma_G(X,Y) &=& \displaystyle \sum_{x \in X,~y \in Y} |x,y|_G,\\
\end{eqnarray*}
\end{minipage}
\begin{minipage}{6cm}
\begin{eqnarray*}
\sigma_G(X,v) &=& \displaystyle \sum_{x \in X} |x, v|_G,\\
\sigma_G(v,X) &=& \displaystyle \sum_{x \in X} |v, x|_G.
\end{eqnarray*}
\end{minipage}
\end{nott}

When the context is clear, we will often omit $G$ in the previous notation. Moreover, by abuse of terminology, if $H$ is a subgraph of $G$, we note  $\sigma_G(H) = \sigma_G(V(H))$.

\subsection{Dealing with bags}\label{subsec:bags}


What we formally need to prove in order to state that, among bags of order $n$, $\hnk$ has a maximum price of symmetrisation is the following theorem.


\begin{thm}\label{thm:bags}
Let $T \in \tnn{k}$ with $k \geqslant 3$ and $n \geqslant 11$,
\begin{eqnarray}
\posm{\sigma}{\hntka} & \leqslant & \posm{\sigma}{\hnkk{n}{k})},\label{eqn:hnka}\\
\posm{\sigma}{\hnkk{n}{k}} & \leqslant & \posm{\sigma}{\hnk}),\label{eqn:hnk}
\end{eqnarray}
for all $a \in A(T)$. Equality \eqref{eqn:hnk} holds if and only if $k = k^*$.
\end{thm}

This way, as previously stated, if a sequence of transformations ends up on some bag $\hntka$, we know it is not extremal since it has a lower price of symmetrisation than $\hnkk{n}{k}$, itself having a lower price of symmetrisation than \hnk.

To prove Theorem \ref{thm:bags}, we need some additional preliminary results. Firstly, we need to prove the following lemma :

\begin{lem}\label{lem:back-max}
Let $T \in \tn$ strongly connected, we have 
\begin{equation*}
\sigma(T) \leqslant \displaystyle \sum_{i=2}^n \binom{i+1}{2},
\end{equation*}
with equality if and only if $T \simeq \bn$.
\end{lem}

This lemma states that among all tournaments, the backward tournament has a maximum price of symmetrisation (since all tournament symmetrisations have the same transmission). The next two lemmas help in the proof of this result.

\begin{lem}
For $n \geqslant 4$, we have
\begin{equation*}
\sigma(\bn) - \sigma(\bnn{n-1}) = \displaystyle \sum_{i=1}^{n} i.
\end{equation*}

\begin{proof}
Indeed, by construction of $\bn$ from $\bnn{n-1}$, we note that the distance matrix of $\bnn{n-1}$ is a sub-matrix of the distance matrix of $\bn$.

Moreover, there exists a unique arrow $(u,v)$ entering the added vertex $v$ (for some $u$). For each other vertex $t$, we have then $|t,v|_{\bn} = |t,u|_{\bnn{n-1}} + 1$ and so $\displaystyle \sum_{t \in \bn} |t,v| = \displaystyle \sum_{i=1}^{n-1} i$.

On the other hand, there are arrows directly linking $v$ to every other vertex $w$ of $T$ but $u$. We have then $|v,w| = 1$ if $w \neq u,v$ and $|v,u| = 2$. We conclude then that $\displaystyle \sum_{w \in \bn} |v,w| = n$.

By putting these three arguments together, we have the expected result.
\end{proof}
\end{lem}

The following lemma is the key of the proof of Lemma \ref{lem:back-max}.

\begin{lem}
Let $T=(V,A) \in \tn$ with $n \geqslant 3$ strongly connected and $T+v \in \tnn{n+1}$ strongly connected a tournament built by adding a vertex $v$ to $T$. We have
\begin{equation*}
\sigma(T+v) - \sigma(T) \leqslant \binom{n+1}{2},
\end{equation*}
with equality if and only if
\begin{enumerate}
\item adding $v$ does not shorten any shortest path in $T$.
\item $\exists v' \in V$ such as $|v',v| = n-1$.
\end{enumerate}
\label{lem:tourn-key}
\end{lem}

\begin{proof}
By hypothesis, we built $T+v$ by adding some arrows such as $T+v$ is a tournament. Let us assume that adding these arrows do not shorten any shortest path in $T$, \ie that $\sigma_{T+v}(V) = \sigma_T(V)$. We have then $\sigma_{T+v}(V \cup \{v\}) = \sigma_T(V) + \underbrace{\displaystyle \sum_{i \in T} |i,v| + |v,i|}_{:=D}$.

Let us look at the quantity $D$. We have $D = n + D'$ for some $D'$, since there are exactly $n$ arrows of type $(v,i)$ or $(i,v)$ in $T+v$. Moreover, $D'$ is maximum when there exists a vertex $v'$ such as $|v',v| = n-1$. In this case, all other vertices of $T+v$ are on the shortest path from $v'$ to $v$. We have then
\begin{equation*}
\begin{array}{lll}
\displaystyle \sum_{i \in T} |i,v| + |v,i| & \leqslant & n + \displaystyle \sum_{i=1}^{n-1} i \\[1.5em]
                                           & = & \displaystyle \binom{n+1}{2}.
\end{array}
\end{equation*}

We note that if there is no vertex $v'$ such as $|v',v| = n-1$, the maximum value of $D'$ cannot be reached (and so, \emph{a fortiori} neither $\binom{n+1}{2}$). On the other hand, if adding $v$ shortened some shortest paths of $T$, we have $\sigma_{T+v}(V) < \sigma_T(V)$ and so the maximum value of $\binom{n+1}{2}$ cannot be reached either. This proves the lemma.
\end{proof}

We are now able to prove Lemma \ref{lem:back-max}.

\begin{proof}[Proof of Lemma \ref{lem:back-max}]
Let $T \in \tn$ with $n \geqslant 3$ strongly connected, let us show that $\sigma(T) \leqslant \displaystyle \sum_{i=2}^n \binom{i+1}{2}$, with equality if and only if $T \simeq \bn$. This way, Lemma \ref{lem:back-max} will be proved.

We proceed by induction on the order $n$ of considered tournaments. The basis holds since $\cnn{3} = \bnn{3}$ is the only strongly connected tournament of order $3$.

Suppose now that $n \geqslant 4$, we have
\begin{equation*}
\begin{array}{lcll}
\sigma(T+v) & \leqslant & \sigma(T) + \displaystyle \binom{n+1}{2}  & \textrm{by Lemma \ref{lem:tourn-key}}\\
            & \leqslant & \displaystyle \sum_{i=2}^{n-1} \binom{i+1}{2} + \binom{n+1}{2} & \textrm{by induction}\\
            &     =     & \displaystyle \sum_{i=2}^{n} \binom{i+1}{2}.     &
\end{array}
\end{equation*}

We note that equality holds in the second inequality if and only if $T \simeq \bn$, by induction. On the other hand, equality holds in the first inequality under Lemma \ref{lem:tourn-key} hypotheses. Putting these two arguments together, we build $T+v$ as $\bnn{n+1}$ from $\bn$. This proves the lemma.
\end{proof}

With this result, we can now finally prove Theorem \ref{thm:bags}. 

\begin{proof}[Proof of Theorem \ref{thm:bags}, Inequation \eqref{eqn:hnka}]
Let $n, k \in \IN$ such as $3 \leqslant k \leqslant n$, we note that for all $T \in \tnn{k}$, we have $\nor{T} \simeq \nor{\knn{k}}$. To prove the lemma, it is then enough to show that for all $T \in \tnn{k}$ and for all $a \in A(T)$, $\sigma(\hnkk{n}{k}) \geqslant \sigma(\hntka)$, with equality if and only if $\hnkk{n}{k} \simeq \hntka$.

On the general bag $\hntka$ illustrated on Figure \ref{sub:lem-bag}, we then have three possible situations :
\begin{itemize}
\item $T$ is strongly connected and adding $P$ does not create any shortcut in $T$,
\item $T$ is strongly connected and adding $P$ creates some shortcuts in $T$,
\item $T$ is not strongly connected.
\end{itemize}
The graph $\hnkk{n}{k}$ is illustrated and labelled the same way in Figure \ref{sub:back-bag}.

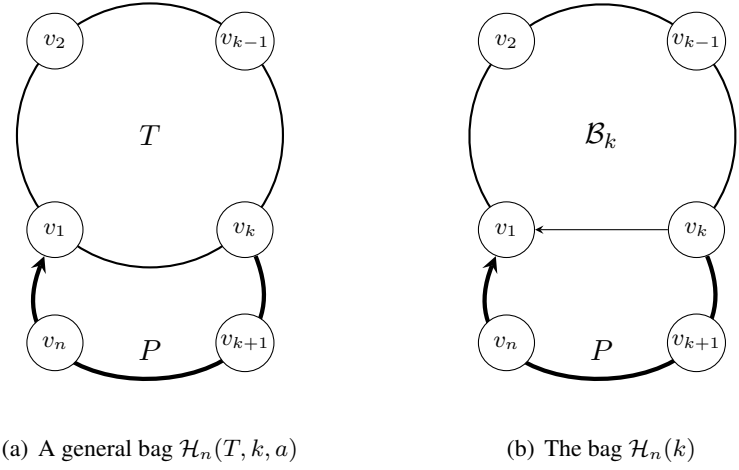
\begin{figure}[!htpd]
\begin{center}
\subfigure[A general bag $\hntka$]{\label{sub:lem-bag}
\begin{tikzpicture}
\draw[thick] (1.25,1.25) circle (1.75);
\draw (1.25, 1.25) node {$T$};
\node[draw, _bigvertex] (v1) at (0, 0) {\footnotesize $v_1$};
\node[draw, _bigvertex] (v2) at (0, 2.5) {\footnotesize $v_2$};
\node[draw, _bigvertex] (v3) at (2.5, 2.5) {\footnotesize $v_{k-1}$};
\node[draw, _bigvertex] (v4) at (2.5, 0) {\footnotesize $v_k$};
\draw[->, >=stealth, ultra thick] (v4) .. controls +(1,-2.5) and +(-1,-2.5).. (v1);
\node[draw, _bigvertex] (v5) at (2.5, -1.5) {\footnotesize $v_{k+1}$};
\node[draw, _bigvertex] (v6) at (0, -1.5) {\footnotesize $v_n$};
\draw (1.25, -1.6) node {$P$};
\end{tikzpicture}
}
\hspace{1cm}
\subfigure[The bag $\hnkk{n}{k}$]{\label{sub:back-bag}
\begin{tikzpicture}
\draw[thick] (2.5, 0) arc (-45:225:1.75) ;
\draw (1.25, 1.25) node {$\bnn{k}$};
\node[draw, _bigvertex] (v1) at (0, 0) {\footnotesize $v_1$};
\node[draw, _bigvertex] (v2) at (0, 2.5) {\footnotesize $v_2$};
\node[draw, _bigvertex] (v3) at (2.5, 2.5) {\footnotesize $v_{k-1}$};
\node[draw, _bigvertex] (v4) at (2.5, 0) {\footnotesize $v_k$};
\draw[->, >=stealth, ultra thick] (v4) .. controls +(1,-2.5) and +(-1,-2.5).. (v1);
\node[draw, _bigvertex] (v5) at (2.5, -1.5) {\footnotesize $v_{k+1}$};
\node[draw, _bigvertex] (v6) at (0, -1.5) {\footnotesize $v_n$};
\draw (1.25, -1.6) node {$P$};
\draw[_arc] (v4) -- (v1);
\end{tikzpicture}
}
\end{center}
\caption{Illustration of the bag $\hntka$ we need to prove non optimal.}
\end{figure}

However, we can handle the three cases at the same time. To simplify notation, we will note $G = \hntka$ and $H = \hnkk{n}{k}$. In every case, we can decompose the transmission computation in the following way :
\begin{equation*}
\sigma(G) = \sigma_G(T) + \sigma_G(\underline{P})+\sigma_G(T,\underline{P}) + \sigma_G(\underline{P},T),
\end{equation*}
where $\underline{P} = P - v_1 - v_k$. We cannot simplify any of the terms more (for instance by substituting $\sigma_G(T)$ by $\sigma(T)$ since we have no information regarding shortest paths in $T$ or $P$ after adding $P$ to $T$. However, when labelling as $v_1, \dots, v_k$ the vertices of $T$, we know the following statement is true :
\begin{equation*}
\sigma(T) = \displaystyle \sum_{\substack{v_i,v_j \in T \\ i < j}} |v_i,v_j|_G + |v_j,v_i|_G.
\end{equation*}

Again, we know pretty nothing about $|v_i,v_j|_G + |v_j,v_i|_G$. However, without loss of generality, since $T$ is a tournament, we can assume $(v_i,v_j) \in T$.  Let us denote $C$ the cycle defined by $(v_i,v_j)$ and a shortest path from $v_j$~to~$v_i$.

But then, on this cycle, we have that $|v_i,v_j| + |v_j,v_i| = 1 + |v_j,v_i| \leqslant n$. Indeed, in the worst case scenario, the shortest path from $v_j$ to $v_i$ crosses all vertices of the graph, and then, crosses $P$. 

Moreover, equality holds in that case and if and only if $|v_1,v_k|=k-1=|v_1,v_i|+1+|v_j,v_k|$, that is, if there is a single shortest path from $v_1$ to $v_k$ of length $k-1$ and $(v_i,v_j)$ lie on this shortest path. On the other hand, if $|v_1,v_k| < k-1$, then $T\not\simeq\bn$ and this equality does not hold. Along with Lemma \ref{lem:back-max}, this allows us to conclude that $\sigma_G(T) \leqslant \sigma_{H}(\bnn{k})$ with equality if and only if $T \simeq \bnn{k}$.

We can use similar arguments to deal with $\sigma(\underline{P})$ and $\sigma_G(T,\underline{P}) + \sigma_G(\underline{P},T)$, allowing us to finally conclude that $\sigma(G) \leqslant \sigma(H)$ with equality if and only if $T \simeq \bnn{k}$.
\end{proof}

We now know that among all bags of order $n$ made from a tournament of order $k$, $\posm{\sigma}{\hnkk{n}{k}}$ is maximum. All we need to do now is to find, for all $n$, what value of $k$ maximises $\posm{\sigma}{\hnkk{n}{k}}$.

To prove the other part of Theorem \ref{thm:bags}, we need the two following lemmas :
\begin{lem}[Trivial]\label{lem:cn}
Let $n \in \IN$, we have
\begin{equation*}
\sigma(\cn) = \frac{n^2(n-1)}{2},
\hspace{2cm}
\sigma(\nor{\cn}) = \left\{
\begin{array}{ll}
\frac{n^3}{4} & \textrm{ if $n$ is even}, \\ [0.5em]
\frac{n(n+1)(n-1)}{4} & \textrm{ otherwise.}
\end{array}
\right.
\end{equation*}
\end{lem}

The next lemma precisely describes computations of $\sigma(\hnkk{n}{k})$ and $\sigma(\nor{\hnkk{n}{k}})$, useful to finally deal with Equation \eqref{eqn:hnk} of Theorem \ref{thm:bags}.

\begin{lem}\label{lem:hnk}
Let $n,k \in \IN$ with $3 \leqslant k < n$, we have
\begin{equation*}
\sigma(\hnkk{n}{k}) = \frac{1}{2} n^3 - \frac{1}{2} n^2 + \frac{1}{2} k (1-k) n + \frac{1}{6} (k-1) (k^2 + 4k + 6),
\end{equation*}
\begin{equation*}
\sigma(\nor{\hnkk{n}{k}}) = \left\{
\begin{array}{ll}
\frac{1}{4}n^3 - \frac{1}{4}(k-2)n^2 - \frac{1}{4} (k-2)(k-6) n & \multirow{2}{*}{\textrm{if $n-k$ is even,}}\\[0.2em]
	+ \frac{1}{4} k (k-2) (k-4) & \\[1em]
\frac{1}{4} n^3 - \frac{1}{4} (k-2) n^2 - \frac{1}{4} (k^2 - 8k + 13) n & \multirow{2}{*}{\textrm{otherwise.}}\\[0.2em]
+ \frac{1}{4} (k-1)(k-2)(k-3) & 
\end{array}
\right.
\end{equation*}
\end{lem}

\begin{proof}
A first easy step in order to prove this lemma is first to compute $\sigma$ for the tournament \bn. We have
\begin{equation*}
\begin{array}{lcll}
\sigma(\bn) & = & \displaystyle \sum_{i=2}^{n} \binom{i+1}{2} & \textrm{by Lemma \ref{lem:back-max}} \\[1.5em]
            & = & \displaystyle \frac{1}{2} \sum_{i=1}^n \big( i^2 \big) + \frac{1}{2} \sum_{i=1}^n \big( i \big) - 1 & \\[1.5em]
	    & = & \frac{1}{6} (n-1) (n^2 + 4n + 6).
\end{array}
\end{equation*}
We are now able to compute $\sigma(\hnkk{n}{k})$. When using the same notation $\underline{P}$ as in the previous proof for a path $P$, we have 
\begin{equation*}
\begin{array}{lcl}
\sigma(\hnkk{n}{k}) &=& \sigma(\bnn{k}) + \sigma(\underline{P}) + \sigma(\underline{P}, \bnn{k}) + \sigma(\bnn{k}, \underline{P})\\[0.5em]
		  &=& \sigma(\bnn{k}) + \displaystyle \sum_{i \in \underline{P}}~\sum_{\substack{j \in \underline{P} \\ i < j}} \big( |i,j| + |j,i| \big) + \sum_{i \in \underline{P}}~\sum_{j \in \bnn{k}} \big( |i,j| + |j,i|\big)\\[2.5em]
		  &=& \sigma(\bnn{k}) + \displaystyle \sum_{i \in \underline{P}}~\sum_{\substack{j \in \underline{P} \\ i < j}} n + \sum_{i \in \underline{P}}~\sum_{j \in \bnn{k}} n\\[2em]
		  &=& \frac{1}{6} (k-1)(k^2 + 4k + 6) + \frac{1}{2} (n-k) (n-k-1) n + k (n-k) n\\[0.5em]
		  &=& \frac{1}{2} n^3 - \frac{1}{2} n^2 + \frac{1}{2} k (1-k) n + \frac{1}{6} (k-1) (k^2 + 4k + 6).
\end{array}
\end{equation*}

The computation of $\sigma(\nor{\hnkk{n}{k}})$ is a little longer and requires to distinguish whether $n-k$ is even or odd. However, in both cases, we can see that
\begin{equation*}
\sigma(\nor{\hnkk{n}{k}}) = \sigma(\nor{\cnn{n-k+2}}) + \sigma(\nor{\knn{k-2}}) + 2 \sigma(\nor{\cnn{n-k+2}}, \nor{\knn{k-2}}).
\end{equation*}
This situation is illustrated in Figure \ref{fig:deltahnk}. In this figure, we can see that no shorter path between two vertices of in $\nor{\cnn{n-k+2}}$ crosses a vertex not in $\nor{\cnn{n-k+2}}$. The same property holds for $\nor{\knn{k-2}}$. Moreover, path going from $\knn{k-2}$ to $\cnn{n-k+2}$ can be split in two parts, illustrated in bold.

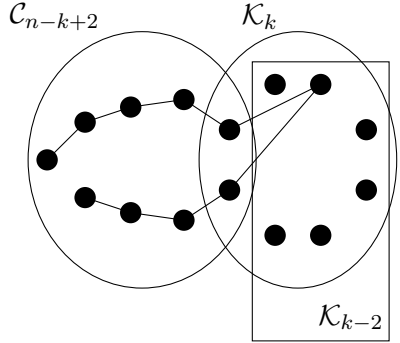
\begin{figure}[!htpd]
\begin{center}
\begin{tikzpicture}
\draw (0, 0) node[_blackv] {} -- (0.5,0.5) node[_blackv] {} -- (1.1,0.7) node[_blackv] {} -- (1.8,0.8) node[_blackv] {}  -- (2.4,0.4) node[_blackv] {} -- (3.6, 1);
\draw (0.5,-0.5) node[_blackv] {} -- (1.1,-0.7) node[_blackv] {} -- (1.8,-0.8) node[_blackv] {}  -- (2.4,-0.4) node[_blackv] {} -- (3.6, 1);
\draw (3, 1) node[_blackv] {};
\draw (3.6, 1) node[_blackv] {};
\draw (4.2, 0.4) node[_blackv] {};
\draw (3, -1) node[_blackv] {};
\draw (3.6, -1) node[_blackv] {};
\draw (4.2, -0.4) node[_blackv] {};
\draw (1.25, 0) ellipse (1.5 and 1.7); 
\draw (3.3, 0) ellipse (1.3 and 1.7); 
\draw (2.7, 1.3) -- (4.5, 1.3) -- (4.5, -2.4) -- (2.7, -2.4) -- cycle;
\draw (4, -2.1) node{$\knn{k-2}$};
\draw (2.8, 1.9) node{$\knn{k}$};
\draw (0.1, 1.9) node{$\cnn{n-k+2}$};
\end{tikzpicture}
\end{center}
\caption{Computation of $\sigma$ in $\nor{\hnkk{n}{k}}$.}
\label{fig:deltahnk}
\end{figure}

\textsc{\underline{Case 1} : $n-k$ is even.}

Using Lemma \ref{lem:cn}, we have 
\begin{equation*}
\begin{array}{lcl}
\sigma(\nor{\hnkk{n}{k}}) &=& \sigma(\nor{\cnn{n-k+2}}) + \sigma(\nor{\knn{k-2}}) + 2 \sigma(\nor{\cnn{n-k+2}}, \nor{\knn{k-2}})\\
			&=& \frac{1}{4} (n-k+2)^3 + (k-2) (k-3) + 2 (k-2) 2 \displaystyle \sum_{i=1}^{\frac{n-k+2}{2}} i \\
			&=& \frac{1}{4} (n-k+2)^3 + (k-2) (k-3) \\[0.5em]
			&& + \frac{1}{2} (k-2)(n-k+2)(n-k+4)\\[0.5em]
			&=& \frac{1}{4}n^3 - \frac{1}{4}(k-2)n^2 - \frac{1}{4} (k-2)(k-6) n + \frac{1}{4} k (k-2) (k-4).
\end{array}
\end{equation*}

\textsc{\underline{Case 2} : $n-k$ is odd.}
 
Again, using Lemma \ref{lem:cn}, we have 
\begin{equation*}
\begin{array}{lcl}
\sigma(\nor{\hnkk{n}{k}}) &=& \sigma(\nor{\cnn{n-k+2}}) + \sigma(\nor{\knn{k-2}}) + 2 \sigma(\nor{\cnn{n-k+2}}, \nor{\knn{k-2}})\\[0.5em]
			&=& \frac{1}{4} (n-k+1)(n-k+2)(n-k+3) + (k-2)(k-3) \\
			&& + 2 (k-2) \left( \displaystyle 2 \sum_{i=1}^{\frac{n-k+1}{2}}\big(i\big) + \frac{n-k+3}{2}\right)\\[1.5em]
			&=& \frac{1}{4} (n-k+1)(n-k+2)(n-k+3) + (k-2)(k-3) \\[0.5em]
			&& + \frac{1}{2} (k-2)(n-k+1)(n-k+3) + (k-2)(n-k+3)\\[0.5em]
			&=& \frac{1}{4} n^3 - \frac{1}{4} (k-2) n^2 - \frac{1}{4} (k^2 - 8k + 13) n \\[0.5em]
			&&+ \frac{1}{4} (k-1)(k-2)(k-3).\\
\end{array}
\end{equation*}
\end{proof}

We are now finally able to prove Theorem \ref{thm:bags}.

\begin{proof}[Proof of Theorem \ref{thm:bags}, Inequation \eqref{eqn:hnk}]
Basically, for this result, we are simply wondering what value of $k$ maximises $\sigma(\hnkk{n}{k})$ for all possible value of $n$. 

We note $k^*$ the value of $k$ such that $\posm{\sigma}{\hnkk{n}{k}}$ is maximum and we analyse this function to determine it. Suppose first that $n-k$ is even. From Lemma \ref{lem:hnk} we have
\begin{equation} \label{eq:pos_even}
\posm{\sigma}{\hnkk{n}{k}} = \frac{-k^3}{12}  + \left( \frac{8-n}{4}\right) k^2 + \left( \frac{3n^2-18n-20}{12} \right) k + \frac{n^3-4n^2+12n-4}{4} \cdot
\end{equation}

The derivative of $\posm{\sigma}{\hnkk{n}{k}}$ with respect to $k$ is 
\begin{equation} \label{eq:dpos_even}
\partial_k \posm{\sigma}{\hnkk{n}{k}} = \frac{-k^2}{4}  + \left( \frac{8-n}{2}\right) k + \frac{3n^2-18n-20}{12},
\end{equation}
and the roots of the derivative are
\begin{equation*}
8-n \pm \sqrt{2 n^2-22n+\frac{344}{6}} \cdot
\end{equation*} 

Because $n \geqslant 11$ and as we are looking for a value of $k^*$ in the range $[2,n-1]$, there is only one root that should be considered (the other being negative). We note this positive root $r_{even}$. Moreover, $r_{even}$ is in the range $[2,n-1]$ since
\begin{equation*}
\begin{array}{lcl}
r_{even} &=& 8-n + \sqrt{2 n^2-22n+\frac{344}{6}} \\
	 &<& 8-n + \sqrt{2 n^2-22n+\frac{363}{6}} \\
	 &=& 8-n + \sqrt{2\left(n - \frac{11}{2}\right)^2}\\
	 &=& n (\sqrt{2} - 1) + 8 - \frac{11\sqrt{2}}{2}\\
	 &\simeq& 0.4142 n + 0.2218.
\end{array}
\end{equation*}
Furthermore, we will note $r = n (\sqrt{2} - 1) + 8 - \frac{11\sqrt{2}}{2}$.

Moreover, $r_{even}$ corresponds to a maximum of the cubic function \eqref{eq:pos_even} since the other root is smaller than $r_{even}$ and the derivate at $ \frac{n}{2}$ (a point bigger than $r_{even}$) is negative. Indeed, when $k = \frac{n}{2}$, the derivative \eqref{eq:dpos_even} becomes
$$
\frac{-3 n ^{2} + 24 n - 80}{48} \cdot
$$

Suppose now that $n-k$ is odd, then,
\begin{equation} \label{eq:pos_odd}
\posm{\sigma}{\hnkk{n}{k}} = \frac{-k^3}{12}  + \left( \frac{8-n}{4}\right) k^2 + \left( \frac{3n^2-18n-29}{12} \right) k + \frac{n^3-4n^2+13n+2}{4} \cdot
\end{equation}
A similar analysis gives only one positive root $r_{odd}$ corresponding to a maximum, that is
$$
r_{odd} = 8-n + \sqrt{2 n^2-22n+\frac{326}{6}} \cdot
$$
Observe that $r_{odd} < r_{even} < r$ and these three values are very closes (they differ only from a constant number of sixth within the square root). More precisely, if $n = 9$,
$$
r - r_{odd} = \sqrt{\frac{49}{2}} - \sqrt{\frac{55}{3}} \simeq 0.668,
$$
and this difference decreases and converges to zero when $n$ grows~:
\begin{equation*}
\begin{array}{cll}
& \displaystyle \lim_{n \rightarrow \infty}  n (\sqrt{2} - 1) + 8 - \frac{11\sqrt{2}}{2} - 8  + n - \sqrt{2n^2 - 22n + \frac{326}{6}} &\\[2em]
=& \displaystyle \lim_{n \rightarrow \infty}  \sqrt{2} n - \frac{11\sqrt{2}}{2}  - \sqrt{2n^2 - 22n + \frac{326}{6}}&\\[2em]
=& \displaystyle \lim_{n \rightarrow \infty} \frac{\sqrt{2}}{2} \left(11 - \frac{163}{6n} \right) - \frac{11\sqrt{2}}{2}\\[1em]
=& 0
\end{array}
\end{equation*}

Observe that $r_{odd}$ is irrational. Thus, $\lfloor r_{odd} \rfloor = \lceil r_{odd} \rceil - 1$. It is also the case for $r_{even}$ and $r$. By convergence, we have either $\lfloor r_{odd} \rfloor = \lfloor r \rfloor$ or $\lfloor r_{odd} \rfloor = \lfloor r \rfloor - 1$. Suppose first that $\lfloor r_{odd} \rfloor = \lfloor r \rfloor$. Then,
$$
\lfloor r_{odd} \rfloor = \lfloor r_{even} \rfloor = \lfloor r \rfloor \textrm{ and } \lceil r_{odd} \rceil = \lceil r_{even} \rceil = \lceil r \rceil,
$$
and the optimal value $k^*$ is clearly
$$
k^* = \max (\posm{\sigma}{\hnkk{n}{\lfloor r \rfloor}}, \posm{\sigma}{\hnkk{n}{\lceil r \rceil}}),
$$
whatever the parity of $n-k$ is. 

Suppose now that $\lfloor r_{odd} \rfloor \neq \lfloor r \rfloor$, \ie $\lceil r_{odd} \rceil = \lfloor r \rfloor$. Since $r - r_{odd}$ converges to zero, we have that $\lceil r_{odd} \rceil$ converges to $r_{odd}$ when $n$ grows, that is, converges to the maximum. In this case, $k^* = \lceil r_{odd} \rceil = \lfloor r \rfloor$. One can argue that for small values of $n$, this could possibly not hold. However, we checked it by computation up to $r - r_{odd} < 10^{-6}$.
\end{proof}

This concludes the proof of Theorem \ref{thm:bags} stating that among all bags of order $n \geqslant 11$, \hnk\ has a maximum price of symmetrisation. 

Conjecture \ref{conj:price-trans} states that when $n \leqslant 10$, the cycle \cn~has a maximum price of symmetrisation and when $n \geqslant 11$, the bag \hnk is extremal. We show now that when $4 \leqslant n \leqslant 10, \posm{\sigma}{\cn} > \posm{\sigma}{\hnk}$, and that the opposite happens when $n \geqslant 11$. 

The polynomial $\IP(n) = \posm{\sigma}{\cn} - \posm{\sigma}{\hnkk{n}{r}}$ can be written 
\begin{equation*}
\left( \frac{5 - 4\sqrt{2} }{12}\right) n^3 + \left( \frac{11\sqrt{2} - 14}{2} \right) n^2 + \left( \frac{944 - 707 \sqrt{2}}{24}\right) n + \frac{2453 \sqrt{2} - 3408}{48}
\end{equation*}

Since $\IP(n)$ has an odd degree, we know it has at least one real root. Moreover, using the Tschirnhaus transformation followed by Scipione del Ferro and Tartaglia method, since the discriminant $ \Delta > 0$, $\IP(n)$ has three distinct real non rational roots. More particularly, we have
\begin{equation*}
\begin{array}{lllllll}
\IP(0) & > &  0 & \wedge & \IP(1) & < & 0, \\
\IP(3) & < &  0 & \wedge & \IP(4) & > & 0, \\
\IP(10) & > &  0 & \wedge & \IP(11) & < & 0. \\
\end{array}
\end{equation*}

Since $\IP(n)$ is continuous, we know these three roots lie in these unit intervals, and we can immediately conclude that $\posm{\sigma}{\cn} > \posm{\sigma}{\hnk}$ when $n \leqslant 10$ and $\posm{\sigma}{\cn} < \posm{\sigma}{\hnk}$ when $n \geqslant 11$.

\subsection{Removing induced $\cnn{2}$}\label{subsec:kill-c2}


As previously stated, this section is devoted to the task of designing transformations able to remove induced $\cnn{2}$ in graphs $G$ while increasing its price of symmetrisation. Indeed, such a configuration does not exist in the assumed extremal graph.

The following notation will be helpful to deal with most of this cases.

\begin{nott}
Let $G \in \gn$ strongly connected, an arrow $a$ of $G$ is said \emph{non critical} if
\begin{enumerate}
\item $G - a$ is strongly connected,
\item $\sigma(G) - \sigma(\nor{G}) < \sigma(G - a) - \sigma(\nor{G - a})$.
\end{enumerate}
\end{nott}

A first obvious transformation is then to remove all non critical arrows from a graph $G$. We will call such a graph \emph{critical}. This leads directly to the next useful lemma.  

\begin{lem}
Let $G \in \gn$ strongly connected and critical. If $G$ has an induced $\cnn{2}$, then both arrows of this $\cnn{2}$ are bridges.
\end{lem}

\begin{proof}
Indeed, if $G$ has an induced $\cnn{2}$ such as one of its arrows $a$ is not a bridge, then $\sigma(G - a) - \sigma(\nor{G - a}) > \sigma(G) - \sigma(\nor{G})$ since $\sigma(G - a) > \sigma(G)$ and $\sigma(\nor{G - a}) = \sigma(\nor{G})$. This means that removing $a$ increases price of symmetrisation, \ie $a$ is non critical, a contradiction by definition of $G$.
\end{proof}

The next lemmas explains how, in most of the time, to define transformations removing $\cnn{2}$ bridges from graphs. This would allow us only to consider critical $\cnn{2}$-free graphs.

\begin{remark}
We note that if some strongly connected graph $G$ as a $\cnn{2}$ bridge between two vertices $x$ and $y$, then this bridge partitions vertices of $G$ in two sets $X$ and $Y$ such as, without loss of generality, $x \in X$, $y \in Y$ and no arrow links $X$ and $Y$ but $(x,y)$ and $(y,x)$. Moreover, both $G[X]$ and $G[Y]$ are strongly connected. We note such a graph $G=(V,A,X,Y,x,y)$. This configuration is illustrated on Figure \ref{fig:bridge}.
\end{remark}

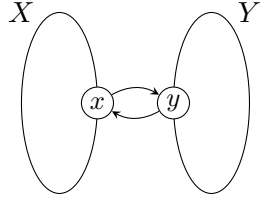
\begin{figure}[!htpd]
\begin{center}
\begin{tikzpicture}
\draw (-0.5, 0) ellipse (0.5 and 1.2);
\draw (1.5, 0) ellipse (0.5 and 1.2);
\draw (-1, 1.2) node {$X$};
\draw (2, 1.2) node {$Y$};
\node[draw, _vertex] (v1) at (0, 0) {\small $x$};
\node[draw, _vertex] (v2) at (1, 0) {\small $y$};
\foreach \i/\j in {1/2, 2/1} {\draw[_arc] (v\i) to[bend left] (v\j);}
\end{tikzpicture}
\end{center}
\caption{A strongly connected graph $G=(V,A,X,Y,x,y)$ partitioned by a double bridge $(x,y)$ and $(y,x)$.}
\label{fig:bridge}
\end{figure}

In the case of such graphs, it is easy to decompose transmission computation in function of transmissions of $G[X]$ and $G[Y]$, as it is in the following lemma. 

\begin{lem}\label{lem:dist-bridge}
Let $G=(V,A,X,Y,x,y) \in \gn$ strongly connected such as $|X| = n_1, |Y| = n_2$. We have
\begin{equation*}
\sigma(G) = \sigma(X) + \sigma(Y) + n_1 \Big(\sigma(Y,y) + \sigma(y,Y)  \Big)  + n_2 \Big(\sigma(X,x) + \sigma(x, X)  \Big) + 2n_1n_2.
\end{equation*}
\end{lem}

\begin{proof}
By separating computations of $\sigma(G)$ in $X$ et $Y$, we have :
\begin{eqnarray*}
\sigma(G) &=&  \sigma(X) + \sigma(Y) + \sigma(X,Y) + \sigma(Y,X) \\
                 &=& \sigma(X) + \sigma(Y) + \displaystyle \sum_{v_1 \in X} \sum_{v_2 \in Y} \Big(|v_1,x| + |x,y| + |y,v_2|\Big) \\
                 &&+ \sum_{v_2 \in Y} \sum_{v_1 \in X} \Big(|v_2,y| + |y,x| + |x,v_1|\Big) \\[1em]
                 &=& \sigma(X) + \sigma(Y) + n_2 \sigma(X,x) + n_1n_2 + n_1 \sigma(y,Y) \\
                 && + n_1 \sigma(Y,y) + n_1n_2 + n_2 \sigma(x,X) \\[1em]
                 &=& \sigma(X) + \sigma(Y) + n_1 \Big(\sigma(Y,y) + \sigma(y,Y)  \Big)  \\
		 && + n_2 \Big(\sigma(X,x) + \sigma(x, X)  \Big) + 2n_1n_2.
\end{eqnarray*}
\end{proof}

The next two lemmas allows us to remove $\cnn{2}$ bridges in most of the cases, the remaining ones are dealt with later. 

\begin{lem}\label{lem:break-c2-no-hyp}
Let $G=(V,A,X,Y,x,y) \in \gn$ strongly connected. Let $H\in \gn$ strongly connected such as $H = G - (y,x) + (y',x')$ with $x' \in X$ and $y' \in Y$ such as $\sigma(x',X)$ and $\sigma(Y,y')$ are maximum. We have
\begin{equation*}
\posm{\sigma}{H} \geqslant \posm{\sigma}{G}
\end{equation*}
with equality if and only if $x=x'$ and $y=y'$. By $\sigma(x',X)$ maximum, we mean $\forall x'' \in X, \sigma(x'', X) \leqslant \sigma(x',X)$. A similar interpretation holds for $\sigma(Y,y')$.
\end{lem}

\begin{proof}
By separating each computations of $\sigma$, we have :
\begin{equation*}
\begin{array}{lcl}
\sigma(G) & = & \sigma_G(X) + \sigma_G(Y) + n_1 \Big(\sigma_G(Y,y) + \sigma_G(y,Y)  \Big) + n_2 \Big(\sigma(_GX,x) + \sigma_G(x, X)  \Big) \\
 		&& + 2n_1n_2  \textrm{~~~~~~~~by Lemma \ref{lem:dist-bridge}}\\[1em]
\sigma(\nor{G}) &=& \sigma_{\nor{G}}(X) + \sigma_{\nor{G}}(Y) + 2 \sigma_{\nor{G}}(X,Y)\\[1em]
\sigma(H) &=& \sigma_G(X) + \sigma_G(Y) + \sigma_H(X,Y) + \sigma_H(Y,X) \\
          &=& \sigma_G(X) + \sigma_G(Y) + \displaystyle \sum_{v_1 \in X} \sum_{v_2 \in Y} \Big(|v_1,x|_G + |x,y|_G + |y,v_2|_G\Big) \\
                 &&+ \displaystyle \sum_{v_1 \in Y} \sum_{v_2 \in X} \Big(|v_2,y'|_G + |y',x'|_H + |x',v_1|_G\Big) \\[2em]
	  &=& \sigma_G(X) + \sigma_G(Y) + n_2 \sigma_G(X,x) + n_1n_2 + n_1 \sigma_G(y,Y) \\
                 && + n_1 \sigma_G(Y,y') + n_1n_2 + n_2 \sigma_G(x',X) \\[0.5em]
          &=& \sigma_G(X) + \sigma_G(Y) + n_1 \Big(\sigma_G(Y,y') + \sigma_G(y,Y)  \Big) \\
		 && + n_2 \Big(\sigma_G(X,x) + \sigma_G(x', X)  \Big) + 2n_1n_2.\\[1em]
\sigma(\nor{H}) &=& \sigma_{\nor{G}}(X) + \sigma_{\nor{G}}(Y) + 2 \sigma_{\nor{H}}(X,Y)
\end{array}
\end{equation*}

We have then :
\begin{equation*}
\begin{array}{lcl}
&& \sigma(H) - \sigma(\nor{H}) - \sigma(G) + \sigma(\nor{G}) \geqslant 0\\[0.5em]
&\ssi& \sigma_G(X) + \sigma_G(Y) + n_1 \Big(\sigma_G(Y,y') + \sigma_G(y,Y)  \Big) + n_2 \Big(\sigma_G(X,x) + \sigma_G(x', X)  \Big) \\
	    && + 2n_1n_2 - \sigma_G(X) - \sigma_G(Y) - n_1 \Big(\sigma_G(Y,y) + \sigma_G(y,Y)  \Big) \\
            && - n_2 \Big(\sigma_G(X,x) + \sigma_G(x, X)  \Big) - 2n_1n_2 - \sigma_{\nor{G}}(X) - \sigma_{\nor{G}}(Y) - 2 \sigma_{\nor{H}}(X,Y) \\
	    && + \sigma_{\nor{G}}(X) + \sigma_{\nor{G}}(Y) + 2 \sigma_{\nor{G}}(X,Y) ~~~~~~~\geqslant 0\\[1em]
&\ssi& n_2 \Big(\sigma_G(x',X) - \sigma_G(x,X) \Big) + n1 \Big(\sigma_G(Y,y') - \sigma_G(Y,y) \Big)\\
	    && + 2 \Big( \sigma_{\nor{G}}(X,Y) - \sigma_{\nor{H}}(X,Y) \Big)~~~~~~~\geqslant 0
\end{array}
\end{equation*}
We note that $\sigma_{\nor{G}}(X,Y) - \sigma_{\nor{H}}(X,Y) > 0$ if and only if $(y',x') \neq (y,x)$ since there are two more arrows in $\nor{H}$ than in $\nor{G}$. Moreover, equality holds if and only if $(y',x') = (y,x)$.

On the other hand, as $x'$ is chosen such as $\sigma(x',X)$ is maximum, we have $\sigma(x',X) - \sigma(x,X) \geqslant 0$. The same argument is valid for $y'$, we hence have $\sigma(H) - \sigma(\nor{H}) \geqslant \sigma(G) + \sigma(\nor{G})$. Equality holds if and only if $(y',x') = (y,x)$.

The property is then verified, with equality if and only if $(y',x') = (y,x)$, that is if and only if $x$ et $y$ are such as $\sigma(x,X)$ and $\sigma(Y,y)$ are maximum.\\
\end{proof}

The following lemma also shows a transformation removing $\cnn{2}$ bridges, while keeping $\posm{\sigma}{G}$ unchanged. It is still useful since combined with Lemma \ref{lem:break-c2-no-hyp}, it allows us to only consider one remaining case of very particular induced $\cnn{2}$. 
In this lemma, $G/(x,y)$ denotes the graph obtained when contracting the arrow $(x,y)$ in $G$.


\begin{lem}\label{lem:bridge-ctr}
Let $G=(V,A,X,Y,x,y) \in \gn$ strongly connected. Let $H\in \gn$ strongly connected such as $H = G/(x,y) + (z,w) + (w,z)$, where $z$ is the vertex obtained when contracting $(x,y)$ in $G$ and $w$ is a new vertex. This situation is illustrated on Figure \ref{sub:contracted}. We have then
\begin{equation*}
\posm{\sigma}{G} = \posm{\sigma}{H}.
\end{equation*}
\end{lem}

\begin{figure}[!htpd]
\begin{center}
\subfigure[The graph $H$]{\label{sub:contracted}
\begin{tikzpicture}
\draw (-0.65, 0) ellipse (0.5 and 1.2);
\draw (0.65, 0) ellipse (0.5 and 1.2);
\draw (-1.1, 1.2) node {$X$};
\draw (1.1, 1.2) node {$Y$};
\node[draw, _vertex] (v1) at (0, 0) {\small $z$};
\node[draw, _vertex] (v2) at (0, -1.5) {\small $w$};
\foreach \i/\j in {1/2, 2/1} {\draw[_arc] (v\i) to[bend left=10pt] (v\j);}
\end{tikzpicture}
}
\subfigure[Distance matrix $M$ of $G$]{\label{sub:mat-g}
\begin{tikzpicture}[scale=0.35]
\draw[color=gray!40, fill=gray!40] (7, 5) -- (7, 10) -- (10,10) -- (10, 5) -- cycle; 
\draw[dotted, thick] (5, 4) -- (5, 5) -- (10, 5);
\draw[dotted, thick] (6, 3) -- (7, 3) -- (7, 10);
\draw (0, 4) -- (10, 4);
\draw (6, 0) -- (6, 10);
\draw (0, 0) -- (10, 0) -- (10, 10) -- (0, 10) -- cycle;
\draw (5.5, 4.5) node {\tiny $0$};
\draw (6.5, 4.5) node {\tiny $1$};
\draw (6.5, 3.5) node {\tiny $0$};
\draw (-1, 11) node {$M$};
\draw (3, 7) node {\tiny $\sigma(X)$};
\draw (3, 2) node {\tiny $\sigma(Y, X)$};
\draw (8, 2) node {\tiny $\sigma(Y)$};
\draw[thick, ->] (-1, 4.5) -- (0, 4.5);
\draw (-1, 4.5) node[left] {\tiny $x$};
\draw[thick, ->] (-1, 3.5) -- (0, 3.5);
\draw (-1, 3.5) node[left] {\tiny $y$};
\draw[thick, ->] (5.5, 11) -- (5.5, 10);
\draw (5.5, 11) node[above] {\tiny $x$};
\draw[thick, ->] (6.5, 11) -- (6.5, 10);
\draw (6.5, 10.9) node[above] {\tiny $y$};
\end{tikzpicture}
}
\subfigure[Distance matrix $M'$ of $H$]{\label{sub:mat-h}
\begin{tikzpicture}[scale=0.35]
\draw[color=gray!40, fill=gray!40] (6, 5) -- (6, 10) -- (9,10) -- (9, 5) -- cycle; 
\draw[dotted, thick] (0,1) -- (0,0) -- (10, 0) -- (10, 10) -- (9, 10);
\draw[dotted, thick] (9, 5) -- (10, 5);
\draw[dotted, thick] (9, 1) -- (10, 1);
\draw[dotted, thick] (9, 1) -- (9, 0);
\draw[dotted, thick] (5, 1) -- (5, 0);
\draw (0, 4) -- (6, 4) -- (6, 10);
\draw (5, 1) -- (5, 5) -- (9, 5);
\draw (0, 1) -- (9, 1) -- (9, 10) -- (0, 10) -- cycle;
\draw (5.5, 4.5) node {\tiny $0$};
\draw (9.5, 0.5) node {\tiny $0$};
\draw (-1, 11) node {$M'$};
\draw (3, 7) node {\tiny $\sigma(X)$};
\draw (7, 3) node {\tiny $\sigma(Y)$};
\draw (7.5, 7.5) node[rotate=60] {\tiny $\sigma$$(X$$-$$x$$,$$Y$$-$$y)$};
\draw[thick, ->] (-1, 4.5) -- (0, 4.5);
\draw (-1, 4.5) node[left] {\tiny $z$};
\draw[thick, ->] (-1, 0.5) -- (0, 0.5);
\draw (-1, 0.5) node[left] {\tiny $w$};
\draw[thick, ->] (5.5, 11) -- (5.5, 10);
\draw (5.5, 11) node[above] {\tiny $z$};
\draw[thick, ->] (9.5, 11) -- (9.5, 10);
\draw (9.5, 11) node[above] {\tiny $w$};
\end{tikzpicture}
}
\end{center}
\caption{Illustration of the graph $H$ and distance matrix comparison.}
\end{figure}
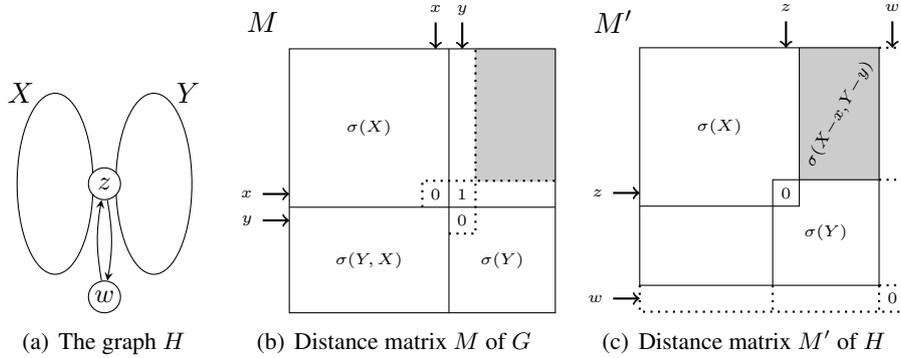

\begin{proof}
Assume the vertices of $X$ are labelled with integers from $1$ to $n_1$ such as $x$ has the label $n_1$ and the vertices of $Y$ are labelled from $n_1 + 1$ to $n$ such that $y$ has the label $n_1 + 1$. In the distance matrix $M$ of $G$, $x$ and $y$ are then two consecutive rows / columns. Moreover, $M$ is partitioned into four submatrices denoting $\sigma(X)$, $\sigma(Y)$, $\sigma(X,Y)$ and $\sigma(Y,X)$, as illustrated in Figure \ref{sub:mat-g}. 

When labelled the same way with $w$ assigned to the highest label, we notice the distance matrix $M'$ of $H$ exhibits a similar structure, illustrated in Figure \ref{sub:mat-h}. The $(x,y)$ bridge of $G$ has been contracted but none of the distances in $X$ and $Y$ have changed. Moreover, shortest paths from a vertex of $X-x$ to a vertex of $Y-y$ have all be shortened by $1$. There are exactly $n_1n_2-n1-n_2+1$ such shortest paths. We can then write $\sigma_H(X-x,Y-y)$ in the following way :

\begin{equation*}
\begin{array}{lcl}
\sigma_H(X-x,Y-y) &=& \sigma_G(X,Y) - (n_1n_2-n1-n_2+1) \\[0.5em]
                  && - \displaystyle \sum_{x' \in X-x} |x',y|_G - \sum_{y' \in Y-y} |x,y'|_G - 1 \\[1.5em]
		  &=& \sigma_G(X,Y) - n_1n_2+n1+n_2-2 \\[0.5em]
                  && - \displaystyle \sum_{x' \in X-x}\Big( |x',x|_G + 1 \Big) - \sum_{y' \in Y-y} \Big( |y,y'|_G + 1 \Big)\\[1.5em]
		  &=& \sigma_G(X,Y) - n_1n_2 - \displaystyle \sum_{x' \in X} |x',x|_G - \sum_{y' \in Y} |y,y'|_G\\[1.5em]
		  && ~~~~~~~~\textrm{since $|x,x|_G = |y,y|_G = 0$}\\
		  &=& \sigma_G(X,Y) - \sigma_G(x,X) - \sigma_G(Y,y) - n_1n_2.
\end{array}
\end{equation*}

We can use similar arguments for $\sigma_H(Y-y,X-x)$ to conclude that 
\begin{equation*}
\sigma_H(Y-y,X-x) = \sigma_G(Y,X) - \sigma_G(X,x) - \sigma_G(y,Y) - n_1n_2.
\end{equation*}

The only remaining distances to consider in $H$ are the distances from $w$ to $X$ and $Y$ (and from $X$ and $Y$ to $w$), as well as the distances from $z$ to $w$ and $w$ to $z$. We can now write $\sigma(H)$ in the following way :

\begin{equation*}
\begin{array}{lcl}
\sigma(H) &=& \sigma_G(X) + \sigma_G(Y) + \sigma_G(X,Y) + \sigma_G(Y,X) - \sigma_G(x,X) - \sigma_G(X,x) \\ 
	  && - \sigma_G(y,Y) - \sigma_G(Y,y) - 2n_1n_2 + \sigma_H(X-x,w) + \sigma_H(w,X-x) \\
	  && + \sigma_H(Y-y,w) + \sigma_H(w,Y-y) + |w,z|_H + |z,w|_H \\[0.5em]
	  &=& \sigma(G) - \sigma_G(x,X) - \sigma_G(X,x) - \sigma_G(y,Y) - \sigma_G(Y,y) - 2n_1n_2 + 2\\
	  && + \displaystyle \sum_{x' \in X-x} \Big( |x',z|_H + |z,x'|_H + 2 \Big) + \sum_{y' \in Y-y} \Big( |y',z|_H + |z,y'|_H + 2 \Big).\\  
\end{array}
\end{equation*}

We note that, in the above formula, we have $|x',z|_H = |x',x|_G$. Similar arguments hold for $|z,x'|$,$|y',z|$ and $|z,y'|$. We have then
\begin{equation*}
\begin{array}{lcl}
\sigma(H) &=& \sigma(G) - \sigma_G(x,X) - \sigma_G(X,x) - \sigma_G(y,Y) - \sigma_G(Y,y) \\
	  && + \displaystyle \sum_{x' \in X-x} \Big( |x',x|_H + |x,x'|_H \Big) + \sum_{y' \in Y-y} \Big( |y',y|_H + |y,y'|_H \Big)\\
	  && - 2n_1n_2 + 2n_1 + 2n_2 - 2\\[0.5em]
	  &=& \sigma(G) - \sigma_G(x,X) - \sigma_G(X,x) - \sigma_G(y,Y) - \sigma_G(Y,y) \\
	  && + \sigma_G(x,X) + \sigma_G(X,x) + \sigma_G(y,Y) + \sigma_G(Y,y) \\
	  && - 2n_1n_2 + 2n_1 + 2n_2 - 2\\
	  && ~~~~~~~~\textrm{since $|x,x|_G = |y,y|_G = 0$}\\[0.5em]
	  &=& \sigma(G) - 2n_1n_2 + 2n_1 + 2n_2 - 2.
\end{array}
\end{equation*}

We note that the above approach to decompose the computation of $\sigma(H)$ is also valid for $\sigma(\nor{H})$. We can then immediately conclude that $\posm{\sigma}{H} = \sigma(G) - \sigma(\nor{G}) = \posm{\sigma}{G}$. 
\end{proof}

We note that the transformation defined in Lemma \ref{lem:bridge-ctr} removes an induced $\cnn{2}$ while creating an other one. The only advantage of using this lemma is then to transform induced $\cnn{2}$ not dealt with previous lemmas into pending induced $\cnn{2}$, a simpler structure. On the other hand, this transformation keeps $\posm{\sigma}{G}$ unchanged. However, if a graph $G$ could be transformed an arbitrary number of times using only Lemma \ref{lem:bridge-ctr}, it would eventually be a tree of induced $\cnn{2}$. Such a graph is symmetric, and has then a null price of symmetrisation. Since Lemma \ref{lem:bridge-ctr} keeps $\posm{\sigma}{G}$ unchanged, it means that $\posm{\sigma}{G} = 0$, and then that $G$ is symmetric as well, and so, not extremal.

In order to completely deal with induced $\cnn{2}$, the only remaining case is a graph $G$ with an induced pending $\cnn{2}$ attached on some vertex $x$ such as both $\sigma(x,X)$ and $\sigma(X,x)$ are maximum. 

\subsection{Contraction - insertion algorithm}\label{subsec:contr-alg}

Regardless of the fact that they are still some induced $\cnn{2}$ we cannot deal with, we were still able to make transformation experiments on $\cnn{2}$-free graphs. On the other hand, when looking at bags structure, motivated by the supposed extremal graph of Conjecture \ref{conj:price-trans}, we notice they all have a possibly long induced path. When dealing with a $\cnn{2}$-free graph $G$, a first intuitive idea is then to lengthen the longest induced path of $G$.

In order to define such a graph transformation, we need to introduce the following notation :
\begin{nott}
Let $G \in \gn$, $P$ the longest induced path of $G$ and $a \in G-P$. We note 
\begin{enumerate}
\item $G/a$ the graph $G$ in which arrow $a$ has been contracted,
\item $G'_a$ the graph $G/a$ in which a vertex has been inserted on $P$.
\end{enumerate}
\end{nott}

We note that if $G$ is strongly connected, then so is $G'_a$ for all $a \in A$. The basic idea behind the graph transformation is then to find the best arrow to contract in $G$ in order to increase price of symmetrisation. More formally, we define this transformation of a graph $G$, noted $T_1(G)$ in the following way :

\begin{algorithm}[!htpd]
\caption{Algorithm $T_1(G)$}
\begin{algorithmic}[1]
\STATE Let $P$ the longest induced path of $G$.
\STATE Let $score(a) := \sigma(G'_a) - \sigma(\nor{G'_a}) - \sigma(G) + \sigma(\nor{G})$.
\STATE $a = \displaystyle \max_{a \in G-P} score(a)$
\IF {$score(a) > 0$}
    \STATE $G := G'_a$
\ENDIF
\end{algorithmic}
\label{alg:contr}
\end{algorithm}

Is is assumed the transformation fails when an arrow of positive score cannot be found. Guided random experiments\footnote{The experiments simply consist of an heuristic graph search, truncated by previously stated lemmas. We also apply these lemmas each time the transformation is applied in order, among other things, to remove non critical arrows and most of induced $\cnn{2}$.} appear to state that such an arrow only exists if $G$ is a bag, or $G$ is a graph with \emph{bunches}. A \emph{bunch} in a digraph is a set of induced path with same start and end points. Samples of experiments leading to graph with bunches are illustrated in Figure \ref{fig:bunch-graph}. Like Conjecture \ref{conj:price-trans}, \DG \cite{DG} was used to automatically find such examples.

\begin{figure}[!htpd]
\begin{center}
\includegraphics[width=4cm]{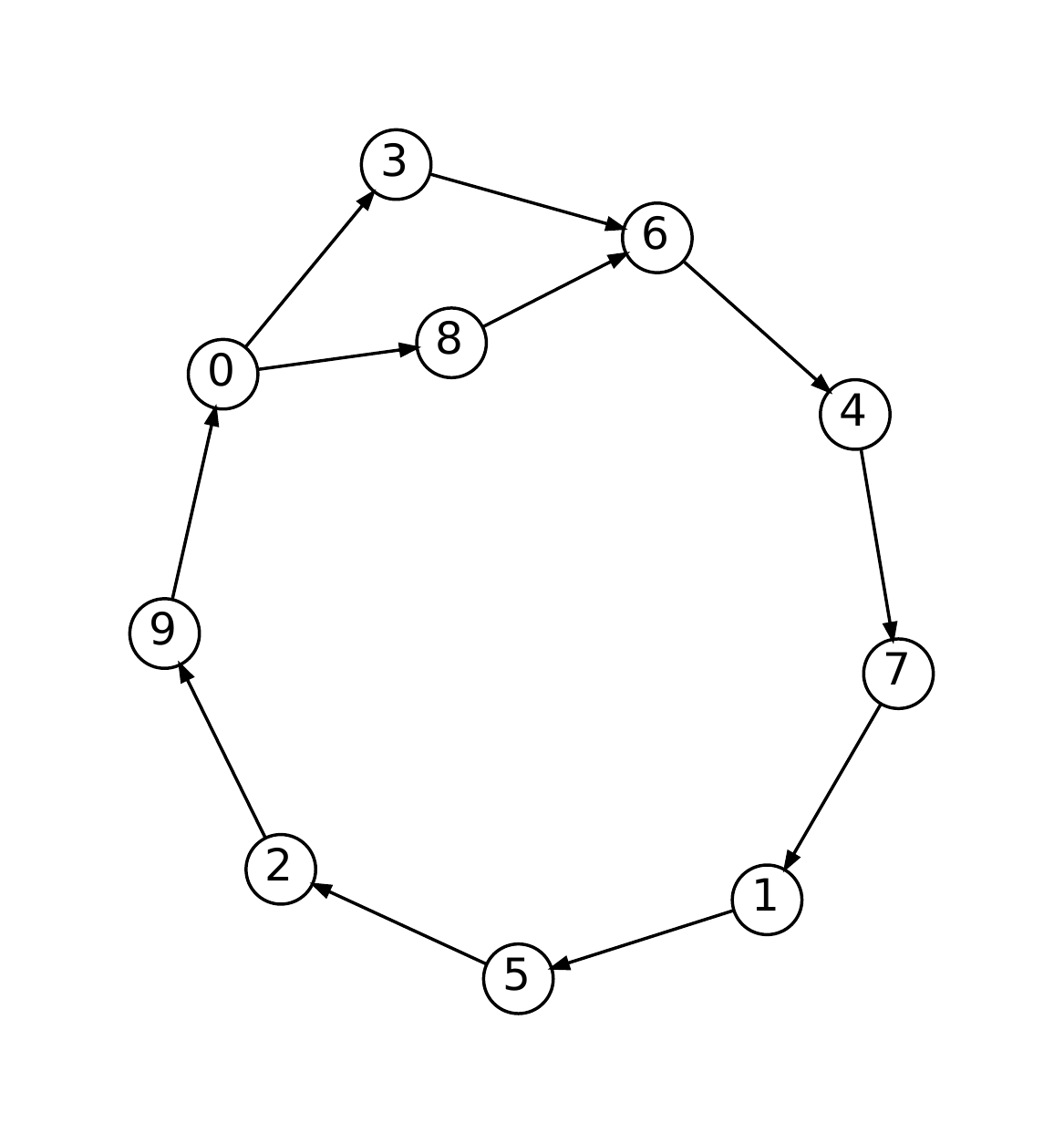} \includegraphics[width=4cm]{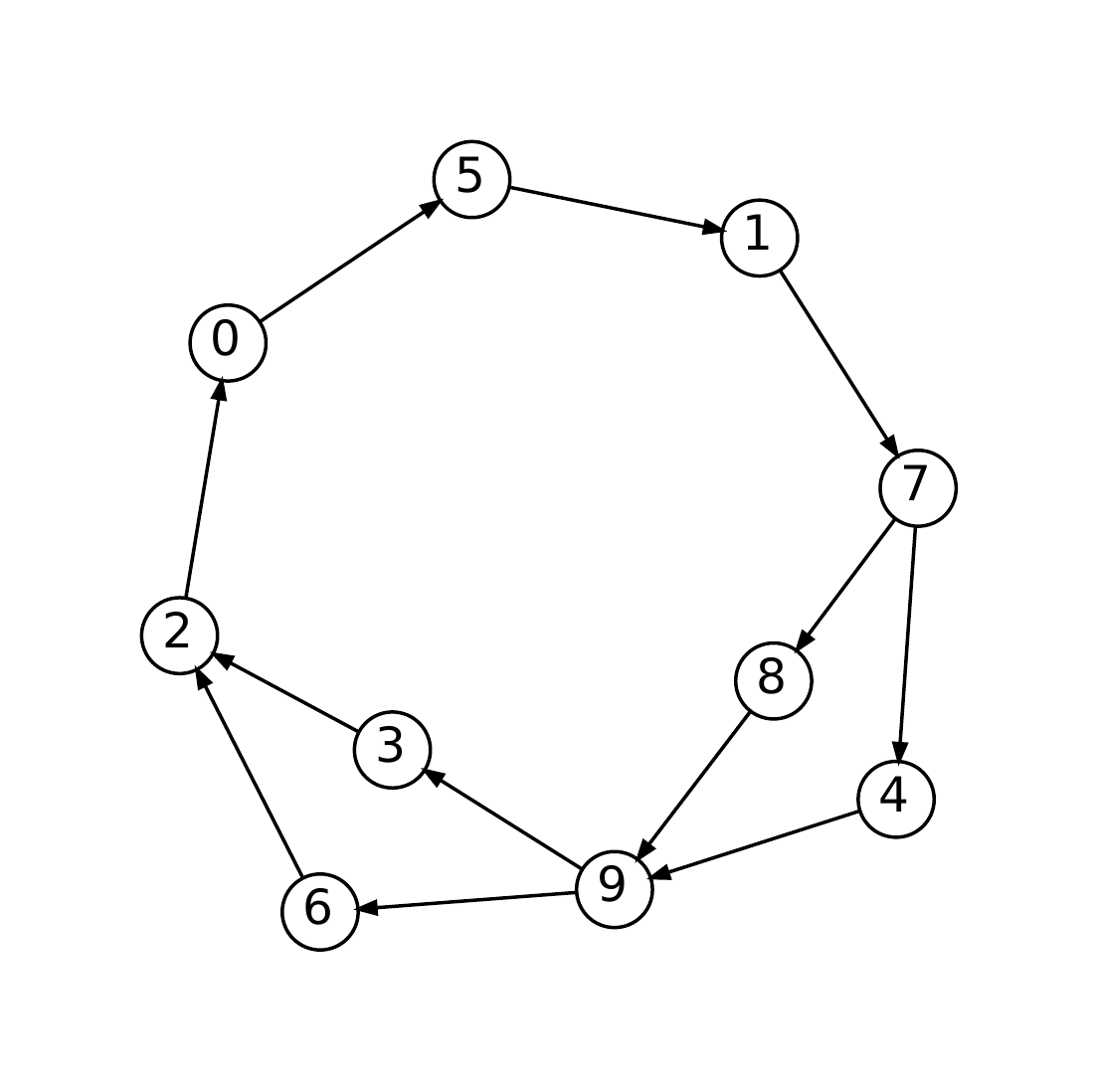} \includegraphics[width=4cm]{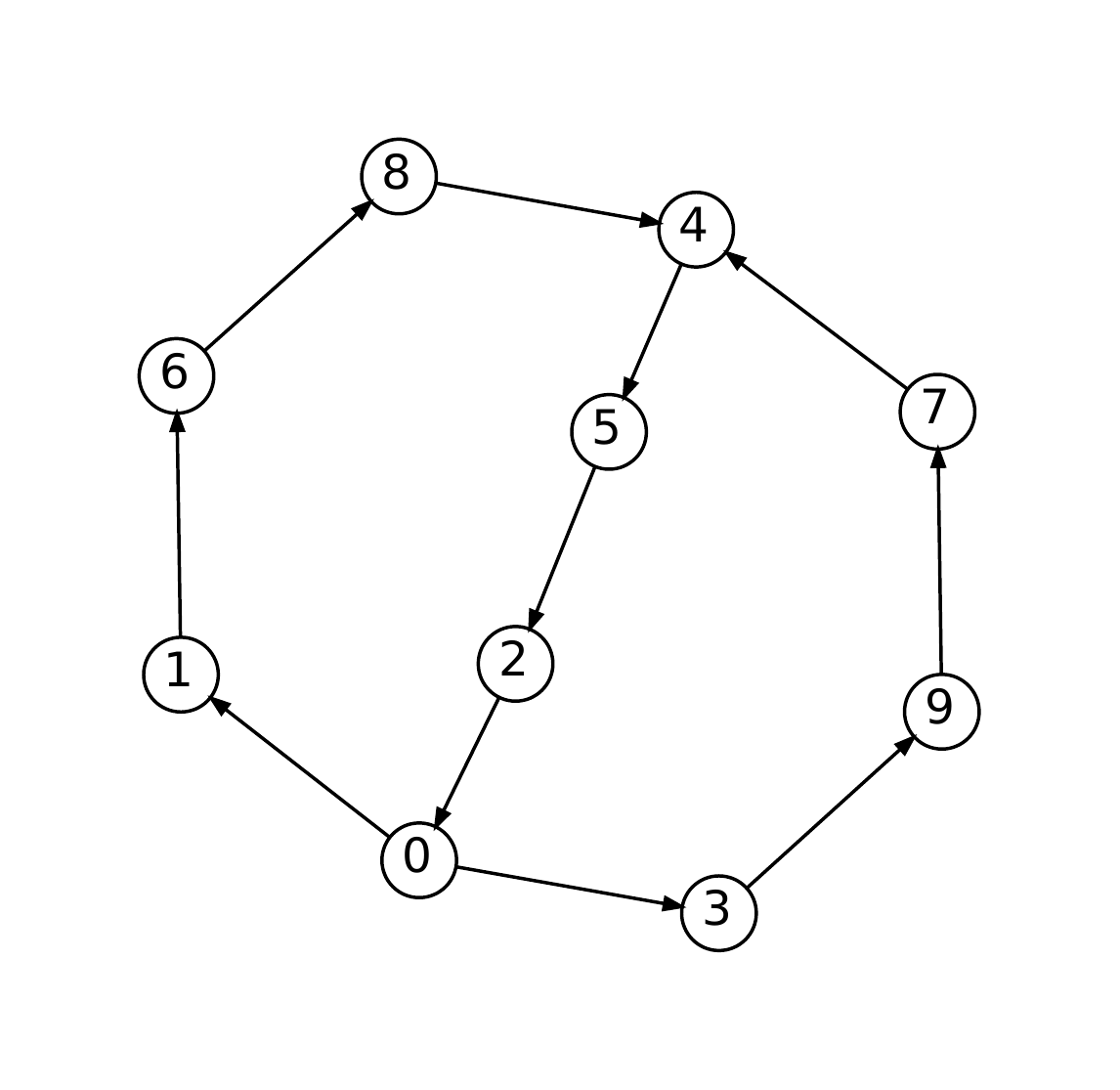}
\end{center}
\caption{Graph with bunches, \ie non bag graphs where $T_1$ fails.}
\label{fig:bunch-graph}
\end{figure}

Since \hnk\ is proved to be the optimal bag, we only have to find another transformation dealing with bunches, \ie find a transformation removing this configuration while increasing price of symmetrisation. However, as for $\cnn{2}$ bridge elimination, we are still unable to deal with bunches.



This result concludes our partial proofs on price of symmetrisation for average distance. As a brief sum up, the following points could be paths of research to prove Conjecture~\ref{conj:price-trans} :
\begin{enumerate}
\item Find a transformation dealing with $\cnn{2}$ bridges that does not fit in Lemma \ref{lem:break-c2-no-hyp}.
\item Prove that the transformation $T_1$ described in Algorithm \ref{alg:contr} always exists when the underlying graph is not a bag or a graph with bunches.
\item Prove that iterative uses of $T_1$ eventually end up on one of these two cases.
\item Find a transformation removing bunches from a graph while increasing price of symmetrisation.
\end{enumerate}

\section{Conclusion}

We have defined the notion of price of symmetrisation, that constitutes a new class of graph invariants for digraphs expressing the gap (or the quotient) of values, for a given invariant, between a digraph and the same digraph that has been symmetrised. We have shown that some extremal questions about the price of symmetrisation are easy (for instance when it concerns the diameter of the domination number) while others are intricate. Indeed, the maximum price of symmetrisation for the average distance and the digraphs achieving it have been conjectured, and although partial results are given, the conjecture remains open for general graphs.

We believe that the notion of price of symmetrisation can lead to other interesting questions since it can be applied to various graph invariants. Moreover, it is a convenient way to express how invariants' values change when one restrict digraphs to be symmetrised (or, conceptually, to be undirected).

\subsection*{Acknowledgements}

The authors would like to thank Alain Hertz, Gilles Caporossi and Hadrien Lepousé for useful discussions about Conjecture~\ref{conj:price-trans}.

\bibliographystyle{acm}
\bibliography{biblio}

\end{document}